\documentclass[runningheads]{llncs}

\usepackage[utf8]{inputenc}
\usepackage{graphicx}
\usepackage{amsmath}
\usepackage{float}
\usepackage{amsfonts}
\usepackage{color}
\usepackage{algorithm}
\usepackage{algpseudocode}
\usepackage{framed}
\usepackage[breaklinks,bookmarks=false]{hyperref}
\usepackage{xifthen}
\usepackage{enumerate}
\usepackage{bibunits}
\newcommand{\bigo}{\mathcal{O}}
\DeclareMathOperator{\hull}{CH}
\DeclareMathOperator{\cost}{cost}
\newcommand{\poly}{\text{poly}}
\newcommand{\GETSP}[1][]{(\ifthenelse{\equal{#1}{}}{V}{#1},T)\mbox{-}\textrm{GETSP}}
\newcommand{\GETSPH}[1][]{(\ifthenelse{\equal{#1}{}}{V}{#1},T,H)\mbox{-}\textrm{GETSP}}

\newcommand{\ETSP}{k\mbox{-}\textrm{ETSP}}

\title{Euclidean TSP with few inner points\\ in linear space}

\author{Pawe\l{} Gawrychowski\inst{1} and Damian Rusak\inst{2}}
\authorrunning{}
\institute{
  Max-Planck-Institut f\"{u}r Informatik, Saarbr\"ucken, Germany
\and
  Institute of Computer Science, University of Wroc{\l}aw, Poland
}

\begin{document}

\maketitle

\begin{abstract}
Given a set of $n$ points in the Euclidean plane, such that just $k$ points are strictly inside the convex hull of the whole set, we
want to find the shortest tour visiting every point. The fastest known algorithm for the version when $k$ is significantly
smaller than $n$, i.e., when there are just few inner points, works in $\bigo(k^{11\sqrt{k}}k^{1.5}n^{3})$ time~[Knauer and Spillner, WG 2006],
but also requires space of order $k^{\Theta(\sqrt{k})}n^{2}$. The best linear space algorithm takes $\bigo(k!kn)$ time~[Deineko, Hoffmann, Okamoto, 
Woeginer, Oper. Res. Lett. 34(1), 106-110]. We construct a linear space $\bigo(nk^2+k^{\bigo(\sqrt{k})})$ time algorithm. The new insight is
extending the  known divide-and-conquer method based on planar separators with a matching-based argument to shrink the instance in every
recursive call. This argument also shows that the problem admits a quadratic bikernel.
\end{abstract}

\section{Introduction}
The traveling salesman problem is one of the most natural optimization questions. Already proven to be NP-hard in the classical book by Garey and Johnson, it remains to be NP-hard even in the most natural Euclidean version~\cite{npcomplete}. A simple $\bigo(2^{n}n^{2})$ dynamic programming can be used to solve the general version, where $n$ is the number of points, but one can do much better by exploiting the additional properties of the Euclidean variant. This was independently observed by Smith~\cite{Smith}, Kann~\cite{kann1992approximability}, and Hwang, Chang, and Lee~\cite{separators}, who all applied a similar reasoning, which we will call the strategy of searching over separators, to achieve an $\bigo(n^{\bigo(\sqrt{n})})$ running time. 
Even though the problem is NP-hard, we might try to construct an algorithm whose running time depends exponentially only on some parameter $k$ of the input instead of the whole $n$. We say that a problem is \emph{fixed-parameter tractable}, if it is possible to achieve a running time of the form $\bigo(f(k)n^{c})$, where $k$ is the parameter.
A closely connected notion is the one of admitting a \emph{bikernel}, which means that we can reduce in polynomial time any its instance to an instance of a different problem, whose size is
bounded by a function of $k$.\footnote{This notion is usually used for decision problem, while we will be working with an optimization question, but this is just a technicality.}
In case of the Euclidean traveling salesman problem, a natural parameterization is to choose $k$ to be the number of inner points, where a point is inner if it lies strictly inside the convex hull of the input. A result of Deineko, Hoffman, Okamoto and Woeginger~\cite{parametrized} is that in such setting $\bigo(2^{k}k^{2}n)$ time is possible (see their paper for an explanation why such parameterization is natural). This was subsequently improved to $\bigo(k^{11\sqrt{k}}k^{1.5}n^{3})$ by Knauer and Spillner~\cite{fasterparametrized}.\footnote{The authors state the result for minimum weight triangulation, but the companion technical report shows that the same strategy works for our problem.} The space consumption of their method (and the previous method) is superpolynomial, as they apply a dynamic programming on $k^{\Theta(\sqrt{k})}n^{2}$ states.

\paragraph{Contribution.}
Our goal is to construct an efficient linear space algorithm. As the previously mentioned exact algorithm for the non-parametrized version~\cite{separators} requires polynomial space, a natural approach is to apply the same strategy. In our case we want the total running time to depend mostly on $k$, though, so we devise a technique of reducing the size of current instance by applying a matching-based argument, which allows us to show that the problem admits a bikernel of quadratic size. By applying the same strategy
of searching over separators on the bikernel, we achieve $\bigo(nk^2+k^{\bigo(k)})$ running time. To improve on that, we extend the strategy by using weighted planar separators, which give us a better handle on how the number of inner points decreases in the recursive calls. The final result is an $\bigo(nk^{2}+k^{\bigo(\sqrt{k})})$ time linear space algorithm.

\paragraph{Overview.}
As in the previous papers, we start with the simple observation that the optimal traveling salesman tour visits the points on the convex hull in the cyclic order. In other words, we can treat subsequent points on the convex hull as the start and end points of subpaths of the whole tour that go only through the inner points. Obviously, no more than $k$ of such potential subpaths include any inner points. We call them \emph{important} and show that we can quickly (in polynomial time) reduce the number of pairs of subsequent points from the convex hull that can create such important subpath to $k^{2}$, and for the remaining pairs we can fix the corresponding edge of the convex hull to be a part of the optimal tour, which shows that the problem admits a bikernel of quadratic size. The reduction shown
in Section~\ref{sec:reduction} is based on a simple (weighted) matching-based argument and works in $\bigo(nk^2 + k^6)$ time and linear space.
The second step is to generalize the \emph{Generalized Euclidean Traveling Salesman Problem}~\cite{separators} as to use the properties of the convex hull more effectively.
In Section~\ref{sec:searching} we modify the strategy of searching over separators, so that its running time depends mostly on the number of inner points. More specifically, we use
the weighted planar separator theorem of Miller~\cite{miller} to prove that there exists a separator whose size is proportional to the 
square root of the number of inner points, irrespectively of the number of outer points. Now if the number of outer points is polynomial, which can
be ensured by extending the aforementioned matching-based reduction, we can iterate over all such separators. Having the separator,
we guess how the solution intersects with it, and recurses on the two smaller subproblems. The separator is chosen so that the number
of inner points decreases by a constant factor in each subproblem, so then assuming the reduction is performed
in every recursive call, we obtain $\bigo(k^{\bigo(\sqrt{k})})$ running time in linear space.

\paragraph{Assumptions.}
We work in the Real RAM model, which ignores the issue of being able to compute distances only up to some accuracy.
By $d(p,q)$ we denote the Euclidean distance between $p$ and $q$. In the rest of the paper, by planar graph we actually mean its fixed
straight-line embedding, as the nodes will be always known points in the plane. Whenever we are talking about sets of points,
we want distinct points, which can be ensured by perturbing them.

\section{The reduction}\label{sec:reduction}

We want to construct an efficient algorithm for a variant of the \emph{Euclidean Traveling Salesman Problem}, called $\ETSP$, in which we are given a set $V$ of $n$ points such that 
exactly $k$ of them lie strictly inside $\hull(V)$, which is the convex hull of the whole set. The algorithm first reduces the 
problem in $\bigo(nk^{2}+k^{6})$ time to an instance of \emph{Generalized Euclidean Traveling Salesman Problem} of size at most $\bigo(k^2)$, and then 
solves the instance in $\bigo(k^{\bigo(\sqrt{k})})$ time. By the size we mean the value of $n+2m$, where $n$ and $m$ are 
defined as below.

\begin{framed}
\noindent \emph{\textbf{Generalized Euclidean Traveling Salesman Problem}} $\GETSP$
\\[5pt]
\noindent Given a set $V = \left\{v_{1},\dots,v_{n}\right\}$ of inner points and a set $T = \left\{(t_{1}, t'_{1}),\dots,(t_{m},t'_{m})\right\}$ of terminal pairs of points, find a set of $m$ paths with the smallest total length such that:
\begin{enumerate}
\item the $i$-th path is built on $(t_{i},t'_{i})$, i.e., it starts from $t_{i}$ and returns to $t'_{i}$,
\item every $v_{i}$ is included in exactly one of these paths,
\end{enumerate}
assuming that in any optimal solution the paths have no self-intersections, and no path intersects other path, except possibly at the ends.
\end{framed}

It is well-known that in an optimal solution to an instance of $\ETSP$ the outer points are visited in order in which they appear on $\hull(V)$
(otherwise the solution intersects itself and can be shortened). Hence we can reduce $\ETSP$ to $\GETSP[V']$ by
setting $V'=V\setminus\hull(V)$ and $T=\left\{ (x_{1},x_{2}),\dots,(x_{n-k},x_{1})\right\}$, where
$\hull(V)=\langle x_{1}, \dots, x_{n-k} \rangle$. As any optimal solution to $\ETSP$ has no self-intersections, the paths in any optimal 
solution to the resulting instance have no self-intersections and do not intersect each other, except possibly at the ends.

We will show that given any instance of $\GETSP$, we can quickly reduce the number of terminal pairs to $\bigo(n^{2})$. A path in a solution to such 
instance is \emph{important} if it includes at least one point from $V$, and \emph{redundant} otherwise. Obviously, a redundant path consists of just one 
edge $(t_{i},t'_{i})$ for some $i$, and the number of important paths in any solution is at most $n$. What is maybe less obvious, we can efficiently 
determine a set of at most $n^{2}$ terminal pairs such that the paths built on the other terminal pairs are all redundant in some optimal solution. To prove 
this, we will notice that every solution to an instance of $\GETSP$ corresponds to a matching, and apply a simple combinatorial lemma. The idea is 
that every important path $\left\langle u_{0},u_{1},\ldots,u_{\ell},u_{\ell+1}\right\rangle$ consists of the middle part $\left\langle u_{1},\ldots,u_{\ell}
\right\rangle$ containing only inner points, and the endpoints $u_{0}=t$, $u_{\ell+1}=t'$ for some terminal pair $(t,t')$. We create a weighted complete 
bipartite graph, where every possible pair of inner points $(u,u')$ corresponds to a left vertex, and every terminal pair corresponds to a right vertex. The 
weight of an edge between $(u,u')$ with $(t,t')$ is $d(t,u)+d(u',t')-d(t,t')$, see Fig.~\ref{fig:paths}(c).

\begin{figure}[t]
\centering
\def \svgwidth{0.8\columnwidth}
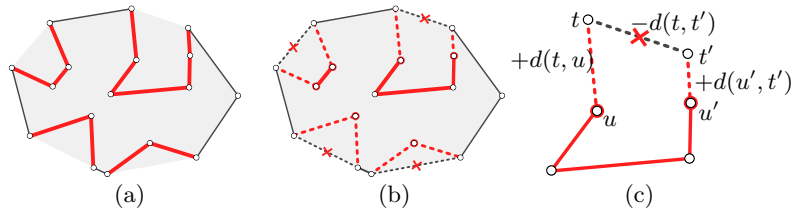
\caption{(a) A solution with important paths marked with thick lines, (b) connecting the inner parts of important paths to form a solution, (c) connecting a single inner part $\langle u, \dots, u' \rangle$ to a terminal pair $(t,t')$ costs $d(t,u)+d(u',t')-d(t,t')$.}
\label{fig:paths}
\end{figure}
 
First we present a simple combinatorial lemma. Given a weighted complete bipartite graph $G=(U\cup V,U\times V,c)$, where $c(u,v)$ is the weight of an edge $(u,v)$,
$\cost(X,Y)$ denotes the weight of a cheapest matching of $X\subseteq U$ to $Y\subseteq V$, if $|X|\leq |Y|$. If $M$ is a matching of $X$ to $Y$, then we denote by $M[X]$ and $M[Y]$ the subsets of $X$ and $Y$ matched by $M$.

\begin{lemma}\label{matching}
Let $G = (U \cup V, U\times V)$ be a weighted complete bipartite graph, where $|U| \leq |V|$.
If $M_{\min}$ is a cheapest matching of $U$ to $V$, then for every $A \subseteq U$ we have $\cost(A,M_{\min}[V])=\cost(A,V)$.
\end{lemma}
\begin{proof}
Assume the opposite, i.e., there is some $A \subseteq U$ such that for any cheapest matching $M$ of $A$ to $V$ we have $M[V]\not\subseteq M_{\min}[V]$. Fix such $A$ and take any cheapest matching $M$ of $A$ to $V$. If there are multiple such $M$, take the one with the largest $|M[V]\cap M_{\min}[V]|$.
Then look at $M\oplus M_{\min}$, which is the set of edges belonging to exactly one of $M$ and $M_{\min}$. It consists of node-disjoint alternating cycles and alternating paths, and the 
alternating paths can be of either odd or even length. Because $M[V]\not\subseteq M_{\min}[V]$, there is a vertex $x\in V$ such that $x\in M[V]\setminus M_{\min}[V]$.
It is clear that there is a (nontrivial) path $P$ starting at $x$, as $x\in M[V]$ but $x\notin M_{\min}[V]$. We want to argue that its length is even.
Its first edge comes from $M$, so if the total length is odd, then
the last edge comes from $M$ as well, so the path ends at a vertex $y\in A$. But all such $y$ are matched in $M_{\min}$, so $P$ cannot end there. Hence it ends at a vertex $y\in 
M_{\min}[V]\setminus M[V]$, and its length is even. Now we consider three cases depending on the sign of $\cost(P)$, which is the total weight of all edges in $P\cap M_{\min}$ minus the 
total weight of all edges in $P\cap M$:
\begin{enumerate}
\item if $\cost(P)>0$ then $M_{\min}\oplus P$ is cheaper than $M_{\min}$, so $M_{\min}$ was not a cheapest matching of $U$ to $V$,
\item if $\cost(P)<0$ then $M\oplus P$ is cheaper than $M$, so $M$ was not a cheapest matching of $A$ to $V$,
\item if $\cost(P)=0$, then $M'=M\oplus P$ is a cheapest matching of $A$ to $V$, and $|M'[V]\cap M_{\min}[V]|>|M[V]\cap M_{\min}[V]|$, so $M$ was not a cheapest matching of $A$ to $V$ with the largest $|M[V]\cap M_{\min}[V]|$ in case of a tie.
\end{enumerate}
Hence there is a cheapest matching $M$ of $A$ to $V$ such that $M[V]\subseteq M_{\min}[V]$.
\qed
\end{proof}

\begin{lemma}\label{matching-alg}
With a read-only constant-time access to a weighted complete bipartite graph $G=(U\cup V,U\times V,c)$, where $|U|\leq |V|$, we can find a cheapest matching of $U$ to $V$ in $\bigo(|U|^{3}+|U||V|)$ time and $\bigo(|V|)$ space.
\end{lemma}

\begin{proof}
Let $p=|U|$ and $q=|V|$. The naive approach would be to find a cheapest matching using $p$ iterations of Dijkstra's algorithm implemented with a Fibonacci heap~\cite{FredmanFibonacci}, which uses $\bigo(p+q)$ space and $\bigo(p((p+q)\log(p+q)+pq))$ total time. We want to smaller time complexity when $p$ is significantly smaller than $q$.

The first straightforward observation is that we can remove all but at most $p^{2}$ nodes from $V$, because we only need to keep, for every $u\in U
$, its $p$ cheapest neighbors from $V$.  By running the aforementioned algorithm on such truncated graph, the total time becomes $\bigo(p^{4})$, but 
the space complexity changes to $\bigo(p^{2})$, which might be larger than $\bigo(p+q)$, so we need an additional idea.

We briefly recap how to use the Dijkstra's algorithm to compute a cheapest matching. We start with an empty matching and iteratively extend the current 
matching $M$ using the cheapest augmenting path. An augmenting path connects an unmatched vertex $u\in U$ with an unmatched vertex $v\in V$ and 
alternates between the nodes of $U$ and $V$. To find the cheapest augmenting path, for every $(u,v)\not\in M$ we create an edge $u\rightarrow v$ with 
a cost of $c(u,v)-\pi_{u}+\pi_{v}$, and for every $(u,v)\in M$ we create an edge $v\rightarrow u$ with a cost of $-c(u,v)+\pi_{u}-\pi_{v}$. The {\it 
potentials} $\pi_{u}$ and $\pi_{v}$ are initially all equal to zero, and then maintained so that the costs of all directed edges are nonnegative, so that we can 
apply the Dijkstra's algorithm to find the cheapest augmenting path. Now consider a single iteration. Let $E_{V}$ be the set of edges incident to the already
matched vertices of $U$. To correctly find the cheapest augmenting path, it is enough to consider, for every $u\in U$, only the cheapest incident edge which does not belong to $E_{V}$. This reduces the complexity of a single iteration to $\bigo(p\log p+|E_{V}|)=\bigo(p^{2})$, assuming that we can quickly extract that cheapest edge for every $u\in U$. To accelerate the extraction, for every $u\in U$ we generate a list $E_{u}$ of $q/p$ cheapest edges incident to $u$ and not belonging to $E_{V}$. The lists are recalculated every $q/p$ iterations. Then, in every iteration, for every $u\in U$ we know that the cheapest incident edge which does not belong to $E_{V}$ belongs to the current $E_{u}$, hence it is enough to run the Dijkstra's algorithm on $|E_{V}+\cup_{u\in U}E_{u}|=\bigo(p^{2}+q)$ edges, which takes $\bigo(p\log p+p^{2}+q)=\bigo(p^{2}+q)$ time and requires $\bigo(p+q)$ space. Because in every iteration exactly one node $v\in V$ becomes matched, recalculating the lists $E_{u}$ every $q/p$ iterations is enough.

Now we analyze how much time do we need to generate every $E_{u}$. We claim that every $E_{u}$ can be found in $\bigo(q)$ time and $\bigo(q/p)$ space. We partition the sequence of all (at most) $q$ edges incident to $u$ and not belonging to $E_{V}$ into blocks of length $q/p$ and process the blocks one-by-one. After processing the first $k$ blocks, we know the $q/p$ smallest elements in the corresponding prefix of the sequence. To process the next block, we take these $q/p$ known smallest elements, add all elements in the current block, and use the linear time median selection algorithm~\cite{median} to select the $q/p$ smallest elements in the resulting set of $2q/p$ numbers. After all blocks are processed, we have exactly the $q/p$ smallest elements of the whole original sequence. The total time complexity is $\bigo(q/p)$ per every block, so $\bigo(q)$ in total, and we clearly need only $\bigo(q/p)$ space.

In every iteration we spend $\bigo(p^{2}+q)$ time to run the Dijkstra's algorithm. Additionally,
every $q/p$ iterations we need $\bigo(q)$ time to recompute the lists $E_{u}$. Hence the total time is $\bigo(p^{3}+pq)$. The space usage is clearly $\bigo(p+q)$.
\qed
\end{proof}

\begin{theorem}\label{GETSP}
Given an instance of $\GETSP$ with $m \geq n^2$, we can find $T_{0}\subseteq T$ of size $m-n^2$, such that there is an optimal solution in which the paths built on pairs from $T_{0}$ are all redundant, in $\bigo(mn^2 + n^6)$ time and $\bigo(m)$ space.
\end{theorem}

\begin{proof}
Let $W = V \times V$ and $G = (W \cup T, W\times T,c)$ be a weighted complete bipartite graph with $W$ and $T$ as the left and right vertices, respectively. 
The weight of an edge connecting $(v,v')$ and $(t,t')$ is defined as
$c((v,v'),(t, t'))=d(t,v)+d(v',t')-d(t, t')$.
Informally, given a path $\left\langle v,\ldots,v'\right\rangle$ consisting of inner points, $c((v,v'),(t,t'))$ is the cost of replacing a redundant path $\left\langle t,t'\right\rangle$ with an important path $\left\langle t,v,\ldots,v',t'\right\rangle$, assuming that we have already taken into the account the length of the inner part $\left\langle v,\ldots,v'\right\rangle$, see Fig.~\ref{fig:paths}(c). Now any solution corresponds to a matching in $G$, because
for every important path $\left\langle t_{i},v_{i},\ldots,v_{i}',t_{i}'\right\rangle$ we can match $(v_{i},v_{i}')$ to $(t_{i},t_{i}')$. More precisely, if we
denote by $i_{1}<\ldots<i_{s}$ the indices of all these important paths and fix their inner parts $\langle v_{i_{j}},\ldots,v'_{i_{j}}\rangle$, then the solution corresponds to a matching of $W'=\{(v_{i_{1}},v'_{i_{1}}),\ldots,(v_{i_{s}},v'_{i'_{s}})\}$ to $T$, and the cost of the solution is equal to the total length of all inner parts plus $\sum_{i}d(t_{i},t'_{i})$ plus the cost of the matching.
In the other direction, any matching of $W'$ to $T$ corresponds to a solution with the given set of inner parts (but possibly different indices of important paths). The cost of that solution is, again, equal to the total length of all inner parts plus $\sum_{i}d(t_{i},t'_{i})$ plus the cost of the matching, so any cheapest matching corresponds to an optimal solution. By Lemma~\ref{matching} we know, that $\cost(W',M_{\min}[T])=\cost(W',T)$, so there always is a cheapest matching of $W'$ to $T$
which uses only the nodes in $M_{\min}[T]$, where $M_{\min}$ is a cheapest matching of $W$ to $T$ in the whole $G$. 
Therefore, we can set $T_{0} =T\setminus M_{\min}[T]$, because there is at least one optimal solution, where the paths built on pairs from such $T_{0}$ are all redundant. Clearly, $|T_{0}| = m-n^2$. Finally, we can use Lemma~\ref{matching-alg} to find a cheapest matching, as we can implement read-only access to any $c((v,v'),(t,t'))$ without explicitly storing the graph, so the total space usage is $\bigo(n)$ and the total time complexity is $\bigo(mn^{2}+n^{6})$ as claimed.
\qed
\end{proof}

\section{Searching over separators}\label{sec:searching}

In this section we briefly recap the method of searching over separators used in~\cite{separators} to solve 
the \emph{Euclidean Traveling Salesman Problem}. At a high level, it is a divide-and-conquer algorithm. We know that an optimal solution
has no self-intersections, hence we can treat it as a planar graph. Every planar graph has a
small simple cycle separator, which is a simple cycle, which can be removed as to split the whole graph into smaller pieces.
Such separator can be used to divide the original problem into smaller subproblems, which are then solved recursively.

\begin{theorem}[Miller~\cite{miller}]\label{Miller}
In any $2$-connected planar graph with nonnegative weights summing up to $1$ assigned to nodes, there exists a simple cycle, called a
simple cycle separator, on at most $2\sqrt{2\left \lfloor{d/2} \right \rfloor N}$ vertices, dividing the graph into the interior and the exterior part,
such that the sum of the weights in each part is at most $\frac{2}{3}$, where $d$ is the maximum face size and $N$ is the number of nodes.
\end{theorem}


Consider an instance of $\GETSP$ and its optimal solution, which by the assumption has no self-intersections, so there exists
a planar graph such that any edge of the solution appears there. Then by Theorem~\ref{Miller} there is a simple cycle
on at most $2\sqrt{2\lfloor d/2\rfloor (n+2m)}$ nodes such that any edge of the solution is either completely outside, completely inside, or lies
on the cycle, and furthermore there are at most $\frac{2}{3}(n+2m)$ points in either the exterior and the interior part.
We want the cycle to be small, so we need to bound $d$. For all inner faces, this can be ensured by simply triangulating them. To ensure
that the outer face is small, we add three enclosing points, see Fig.~\ref{fig:three}.

\begin{figure}[t]
\centering
\def \svgwidth{0.65\columnwidth}
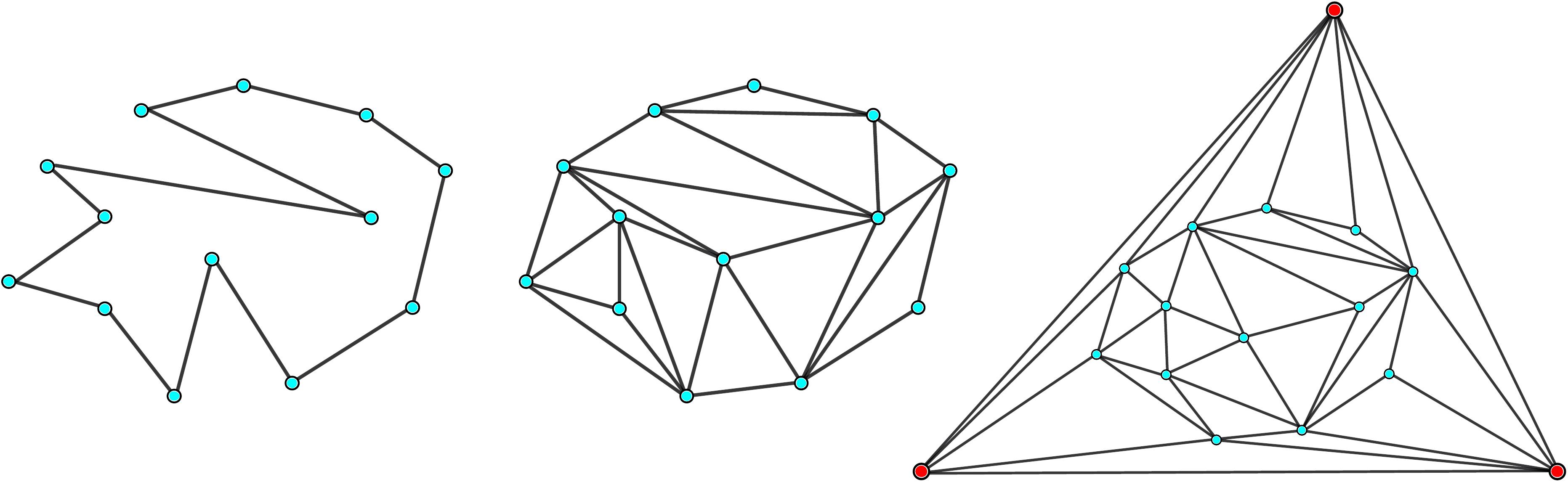
\caption{A solution and its corresponding graph with $d=9$, after adding the three enclosing points $d=3$.}
\label{fig:three}
\end{figure}

Now consider how the paths in the solution intersect with the simple cycle separator. Each path is either completely outside, completely inside, or
intersects with one of the nodes of the separator. Any such intersecting path can be partitioned into shorter 
subpaths, such that the endpoints of the subpaths are either the endpoints of the original paths or the nodes of the separator,
and every subpath is outside or inside, meaning that all of its inner nodes are completely outside or completely inside. This suggest that we can
create two smaller instances of $\GETSP$ corresponding to the interior and the exterior part of the graph, such that the solutions of these two
smaller subproblems can be merged to create the solution for the original problem, see Fig.~\ref{fig:tourcycle}.

\begin{figure}[b]
\centering
\def \svgwidth{0.94\columnwidth}
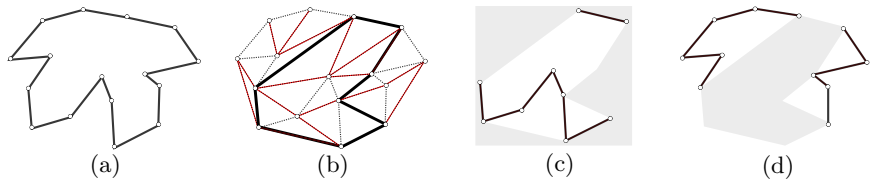
\caption{(a) A solution, (b) the triangulated planar graph and its simple cycle separator, (c) a solution to the interior subproblem, (d) a solution to the exterior problem.}
\label{fig:tourcycle}
\end{figure}

Of course we don't know the solution, so we cannot really find a simple cycle separator in its corresponding triangulated planar graph.
But the size of the separator is at most $c\sqrt{n+2m+3}$ for some constant $c$, so we can iterate over all possible simple cycles of such length,
and for every such cycle check if it partitions the instance into two parts of sufficiently small sizes. The number of cycles is at most $c\sqrt{n+2m+3} \binom{n+2m+3}{c\sqrt{n+2m+3}}(c\sqrt{n+2m+3})!$, which is $\bigo((n+2m)^{\bigo(\sqrt{n+2m})})$.

Similarly, because we don't know the solution, we cannot check how it intersects with our simple cycle separator. But, again, we can iterate over
all possibilities. To bound the number of possibilities, we must be a little bit more precise about what intersecting with the separator means. We create
a number of new terminal pairs. Every node of the separator appears in one or two of these new terminal pairs. 
Additionally, the new terminal pairs might
contain some of the original terminal points, under the restriction that for any original terminal pair, either none of its points are used in the
new terminal pairs, or both are (and in the latter case, we remove the original terminal pair).
Additional, there cannot exist a sequence of new terminal pairs creating a cycle, i.e.,
$(p_{1},p_{2}), \ldots, (p_{\ell-1},p_{\ell}), (p_{\ell},p_{1})$ with $\ell\geq 3$.
Then, for every new terminal pair $(p,p')$, we decide if its path lies fully within the exterior or the interior part
(if it directly connects two consecutive points on the cycle, we can consider it as
belonging to either part). Notice that if $p$ is one of the original terminal points, and $p'$ is a new terminal point, then the corresponding path
lies fully within the part where $p$ belongs to. One can see that such a choice allows us to partition the original problem into two smaller
subproblems, so that their optimal solutions can be merged to recover the whole solution, and that the subproblems are smaller instances
of $\GETSP$. Hence iterating over all choices and choosing an optimal solution in every subproblem allows us to find an optimal solution for the
original instance. To bound the number of choices, the whole process can be seen as partitioning the nodes of the separator into
ordered subsets, selecting two of the original terminal points for every of these subsets, and finally guessing, for every two nodes subsequent in
one of the subsets, whether the path connecting them belongs to the exterior or the interior part. We must also check if it holds that for
any original pair $(p,p')$ it holds that either none of its points was selected, or both of them were, but even without this last easy check
the number of possibilities is bounded by $B_{c\sqrt{n+2m+3}}(c\sqrt{n+2m+3})!\binom{2m}{2c\sqrt{n+2m+3}} 2^{c\sqrt{n+2m+3}}$,
where $B_{s}$ is the $s$-th Bell number. This is, again, $\bigo((n+2m)^{\bigo(\sqrt{n+2m})})$.

The algorithm iterates over all separators and over all possibilities of how the solution intersects with each of them. For each choice, it recurses on
the resulting two smaller subproblems, and combines their solutions. Even though we
cannot guarantee that all optimal solutions in these subproblems have no self-intersections, any optimal solution
to the original problem has such property, so for at least one choice the subproblems will have such property, which is enough for the correctness.
Because the size of every subproblem is at most $b=\frac{2}{3}(n+2m)+c\sqrt{n+2m+3}$, the recurrence for the total running time is
$T(n+2m) = \bigo((n+2m)^{\bigo(\sqrt{n+2m})}) \cdot 2T(b)$.
For large enough $n+2m$, we have that $b \leq \frac{3}{4}(n+2m)$, and the recurrence solves to
$T(n+2m)=\bigo((n+2m)^{\bigo(\sqrt{n+2m})})$. The space complexity is linear, because we only need to generate the subproblems,
which requires iterating over all subsets and all partitions into ordered subsets, and this can be done in linear space.

\section{$\GETSPH$}

To extend the divide-and-conquer algorithm described in the previous section, we need to work with a slightly extended version of $\GETSP$, which
is more sensitive to the number of terminal pairs such that both points belong to the convex hull. We call the extended version $\GETSPH$, and define
its size to be $n+2m+2\ell$.
Given an instance of $\ETSP$, we can reduce the problem to solving an instance of $\GETSPH$ with $|V|=k$, $|T|=0$, and $|H|=n-k$.

\begin{framed}
\noindent \emph{\textbf{Generalized\,Euclidean\,Traveling\,Salesman\,Problem}}\,$\GETSPH$
\\[5pt]
\noindent Given a set $V = \left\{v_{1},\dots,v_{n}\right\}$ of inner points, a set $T = \left\{(t_{1}, t'_{1}),\dots,(t_{m},t'_{m})\right\}$ of terminal pairs of 
points, and a set $H = \left\{(h_{1}, h'_{1}),\dots,(h_{\ell},h'_{\ell})\right\}$ of hull pairs of points, where for any $i$ the point $h_{i}$ and $h'_{i}$ are 
neighbors on the convex hull of the set of all points\footnotemark, find a set of $m+\ell$ paths with the smallest total length such that:
\begin{enumerate}
\item the $i$-th path is built on $(t_{i},t'_{i})$, for $i=1,2,\ldots,m$,
\item the $m+i$-th path is built on $(h_{i},h'_{i})$, for $i=1,2,\ldots,\ell$,
\item every $v_{i}$ is included in exactly one of these paths,
\end{enumerate}
assuming that in any optimal solution the paths have no self-intersections, and no path intersects other path,
except possibly at the ends.
\end{framed}
\footnotetext{Other points given in the input might or might not lie on the convex hull.}

We will show that if $\ell=\poly(n)$, then $\GETSPH$ can be solved in $\bigo((n+2m)^{\bigo(\sqrt{n+2m})})$ time and linear space using an extension of the 
method from the previous section. Combined with Theorem~\ref{GETSP}, this gives an $\bigo(nk^2+k^{\bigo(\sqrt{k})})$ time and linear space 
solution for $\ETSP$. First we extend Theorem~\ref{GETSP}.

\begin{lemma}\label{GETSPH}
Take an instance of $\GETSPH$ with $n = |V|$, $m = |T|$, and $\ell = |H|$. If $m+\ell > n^2$ then in $\bigo((m+\ell)n^{2} + n^6)$ time and
$\bigo(m+\ell)$  space we can find $T_{0}\subseteq T$ and $H_{0}\subseteq H$ such that $|T_{0}| + |H_{0}| = m+\ell-n^2$ and there is an optimal
solution in which the paths built on pairs from $T_{0} \cup H_{0}$ are all redundant.
\end{lemma}
Now applying the divide-and-conquer method described in the previous section directly together with the above lemma gives us a running time of
$\bigo((n+2m+2\ell)^{\bigo(\sqrt{n+2m+2\ell)}})=\bigo(n^{\bigo(n)})$, and we want to improve on that to get $\bigo(n^{\bigo(\sqrt{n})})$.

Recall that the recursive method described in the previous section iterates over simple cycle separators. Because now the (unknown) graph is on
$n+2m+2\ell$ vertices, the best bound on the length of the separator that we could directly get from Theorem~\ref{Miller} is $c\sqrt{n+2m+2\ell+3}$,
which is too large. But say that we can show that there exists a simple cycle separator of length $\bigo(\sqrt{n+2m+2})$, such that the value of $n+2m$ 
decreases by a constant factor in both parts. Iterating over all such simple cycle
separators takes $\bigo((n+2m+2\ell)^{\bigo(\sqrt{n+2m})})$ time, and iterating over all possibilities of how the separator intersects with the solution
then takes $B_{\bigo(\sqrt{n+2m})}\bigo(\sqrt{n+2m}!)\binom{n+2m+2\ell}{\bigo(\sqrt{n+2m})} 2^{\bigo(\sqrt{n+2m})}$ time. All in all, the
total number of possibilities becomes $\bigo((n+2m+2\ell)^{\bigo(\sqrt{n+2m})})$, which assuming that $\ell=\poly(n)$ is 
$\bigo((n+2m)^{\bigo(\sqrt{n+2m})})$. Applying this reasoning in a recursive manner as in the previous section results in Algorithm~\ref{algorithm}.
Compared to the algorithm from the previous section, the changes are as follows:
\begin{enumerate}
\item we reduce the number of terminal and hull pairs using Lemma~\ref{GETSPH} in line~\ref{line:reduce},
\item we add just two enclosing points (instead of three) in line~\ref{line:enclosing},
\item  when forming the subproblems in line~\ref{line:form}, we might connect both some terminal points and some hull points with the nodes
of the separator, and in the latter case, the new pair always becomes a terminal pair in the subproblem.
\end{enumerate}
If $\ell=\poly(n)$ in the original instance, then we can maintain such invariant in all recursive calls without increasing
the running time, because the (polynomial) cost of the reduction in a subproblem can be charged to its parent. Therefore, because the value of
$n+2m$ decreases by a constant factor in both subproblems, the total time is $\bigo((n+2m)^{\bigo(\sqrt{n+2m})})$ by the same recurrence as previously.

\begin{algorithm}[t]
\caption{For solving $\GETSPH$.}
\label{algorithm}
\begin{algorithmic}[1]
\algtext*{EndIf}
\algloopdefx{NoEndIf}[1]{\textbf{if} #1 \textbf{then}}
\algloopdefx{NoEndIfD}[1]{\indent \indent \textbf{if} #1 \textbf{then}}
\algloopdefx{NoEndForAll}[1]{\textbf{for each} #1 \textbf{do}}
\algloopdefx{NoEndForAllD}[1]{\indent \textbf{for each} #1 \textbf{do}}
\If{$V = \emptyset$} \textbf{return} all edges directly connecting the pairs in $T \cup H$
\EndIf
\If{$m+\ell > n^2$} \label{alg:reducestart} \label{line:reduce}
\State Apply Lemma~\ref{GETSPH} to find $T_{0}$ and $H_{0}$. \Comment{$\bigo((m+\ell)n^{2}+n^{6})$}
\State Directly connect the redundant pairs in $T_{0}\cup H_{0}$.
\EndIf
\State Add two enclosing points $I_{1}$ and $I_{2}$. \label{line:enclosing}
\NoEndForAll {ordered subset $C$ of all points with $|C|\leq c\sqrt{n+2m+2}$}
\State \:\:Check if $C$ forms a simple cycle.
\State \indent Check if there are at most $\frac{2}{3}(n+2m)$ inner and terminal points in either part.
\NoEndForAllD {possibility of how the solution intersects with $C$}
\State \:\:\indent Form the exterior subproblem and the interior subproblem. \label{line:form}
\State \indent\indent Recursively solve the exterior subproblem.
\State \indent\indent Recursively solve the interior subproblem.
\State \indent\indent Combine the solutions for the subproblems and update the best solution.
\State \textbf{return} the best solution found in the whole process.
\end{algorithmic}
\end{algorithm}

Now the goal is to prove that it is enough to consider simple cycle separators of length $c\sqrt{n+2m+2}$. To this end, we will prove that
there exists a planar graph with the following properties:
\begin{enumerate}[(a)]
\item its set of nodes includes all inner and terminal points together with the two enclosing points,
and possibly some hull points,\label{prop:1}
\item any edge from the solution is either an edge in the graph, or lies within one of its faces,\label{prop:2}
\item all of its faces are of size at most $4$ and its size is $\bigo(n+2m)$.\label{prop:3}
\end{enumerate}
If such a graph exists, then by Theorem~\ref{Miller} it has a simple cycle separator of size $\bigo(\sqrt{n+2m})$ due to (\ref{prop:3}).
Furthermore, by assigning equal weights summing up to one to all inner and terminal
points, which by (\ref{prop:1}) are nodes of the graph, and zero weights to the remaining nodes, we get a simple cycle separator which,
by (\ref{prop:2}), divides the original problem into two subproblems, such that the optimal solution to the subproblems can be combined to form
an optimal solution to the original problem, and there are at most $\frac{2}{3}(n+2m)$ inner and terminal points in every subproblem, so
Algorithm~\ref{algorithm} is correct. Before we show that such a graph exists, we provide the details of how to choose the enclosing points.

\begin{lemma}\label{enclosing}
For any set of points $A$ we can find two enclosing points $I_{1}$, $I_{2}$, lying outside $\hull(A)$, and two nodes of $\hull(A)$ called $v_{up}, v_{down}$, 
such that:
\begin{enumerate}
\item all points of $\hull(A)$ between $v_{up}$ and $v_{down}$ (clockwise) lie inside $\bigtriangleup I_{2}v_{up}v_{down}$, and 
all points of $\hull(A)$ between $v_{down}$ and $v_{up}$ lie inside $\bigtriangleup I_{1}v_{up}v_{down}$,
\item for any point $w$ of $\hull(A)$ between $v_{up}$ and $v_{down}$, $I_{1}w$ has no common point
with $\hull(A)$ except for $w$, and for any $w$ between $v_{down}$ and $v_{up}$, $I_{2}w$ has no common point
with $\hull(A)$ except for $w$.
\end{enumerate}
\end{lemma}

\begin{figure}[b]
\centering
\def \svgwidth{0.7\columnwidth}
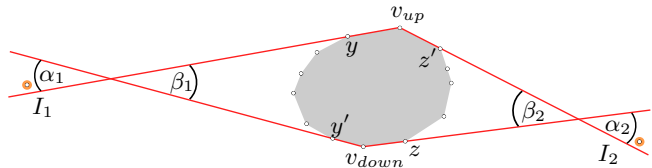
\caption{Choosing the enclosing points $I_{1}$ and $I_{2}$.}
\label{fig:zones}
\end{figure}

\begin{proof}
If $A$ contains less than two points, any two $I_{1}$ and $I_{2}$ are fine. For any larger $A$,
we can find two distinct parallel lines $k_{1}$ and $k_{2}$, such that each of them has exactly one common
point with $\hull(A)$. Call these common points $v_{up}$ and $v_{down}$, respectively. Let $y,z'$ be the neighbors of $v_{up}$ on $\hull(A)$ and $z,y'$ be 
neighbors of $v_{down}$, such that $y,y'$ are on the other side of $v_{up}v_{down}$ than $z,z'$. Consider the angle $\beta_{1}$ obtained by extending 
segments $v_{up}y$ and $v_{down}y'$, and the angle $\beta_{2}$ obtained by extending $v_{up}z'$ and $v_{down}z$. Finally, let $\alpha_{1}$ and $
\alpha_{2}$ be angles vertically opposite to $\beta_{1}$ and $\beta_{2}$, respectively, see Fig.~\ref{fig:zones}. Now we choose $I_{1}$ as any
point strictly inside $\alpha_{1}$ and $I_{2}$ as any point strictly inside $\alpha_{2}$. We must show that for such a choice both properties
hold. Because of the symmetry, it is enough to prove the first part of each of them.

\begin{enumerate}
\item
One of the properties of a convex hull is that all of its points lie inside the intersection of the halfplanes, which are defined by its segments.
The intersection of the halfplanes defined by the segments $v_{up}y$ and $v_{down}y'$ is precisely $\beta_{1}$. All points between $v_{up}$ and
$v_{down}$ in the counterclockwise order lie on the same side of $v_{up}v_{down}$ as $I_{1}$. Therefore they all lie inside the part of $\beta_{1}$ bounded 
by the segment $v_{up}v_{down}$. Due to our choice of $I_{1}$ this part lies inside $\bigtriangleup I_{1}v_{up}v_{down}$, and so the first propery holds.
\item
Assume the opposite, i.e., the segment $I_{1}w$ has a common point with $CH(A)$ other that $w$, call it $u$. Clearly, $\bigtriangleup wv_{up}v_{down}$ 
lies  inside $\bigtriangleup uv_{up}v_{down}$ and due to the convexity of the hull all points strictly inside $\bigtriangleup uv_{up}v_{down}$ are strictly 
inside  $CH(A)$ as well. But that is in contradiction with $w$ being a node of $CH(A)$, and so the second property holds.
\qed
\end{enumerate}
\end{proof}


We say that $(U,H,S)$ is a \emph{hull structure} if:
\begin{enumerate}
\item $U$ and $H$ are two sets of points in the plane with $H \subseteq \hull(U \cup H)$,
\item $S$ is a collection of segments connecting the points in $U\cup H$ such that no segment intersects other segment, except possibly at the ends,
\item any point from $U\cup H$ is an endpoint of at most two segments in $S$,
\item every segment in $S$ connecting two points from $H$ lies on $\hull(U \cup H)$.
\end{enumerate}
One can easily see that any optimal solution to an instance of $\GETSPH$ corresponds to a hull structure $(U,H,S)$, where $U$ consists of all inner
and terminal points, $H$ contains all hull points, and $S$ is a collection of segments constituting the paths. Furthermore, for any hull structure the following holds.


\begin{lemma}\label{small_triang}
If $(U,H,S)$ is a hull structure, and $I_{1}, I_{2}$ are the points enclosing $U\cup H$, then there exists a planar graph, such that:
\begin{enumerate}
\item the nodes are all points from $U\cup\{I_{1},I_{2}\}$ and possibly some points from $H$,
\item any segment from $S$ is either an edge of the graph, or lies within its face,
\item all of its faces are of size at most $4$,
\item the size of the graph is $\bigo(|U|)$.
\end{enumerate}
\end{lemma}

\begin{proof}
The enclosing points $I_{1}, I_{2}$ are defined by applying Lemma~\ref{enclosing} on $U\cup H$, same for $v_{up}$ and $v_{down}$. The subsets
of $U\cup H$ on the same side of the line going through $v_{up}v_{down}$ as $I_{1}$ and $I_{2}$ will be called $V_{1}$ and $V_{2}$, respectively.
The subset of $S$ containing all segments with at least one endpoint in $U$ will be called $S'$. Because any point of $U$ is an endpoint of at most
two segments, $|S'|=\bigo(|U|)$. We define $U'$ to be the whole $U$ together with the points of $H$ which are an endpoint of some segment in $S'$, and
create the first approximation of the desired planar graph using $U'$ as its set of nodes, and $S'$ as its set of edges. We triangulate this
planar graph, so that its inner faces are of size $3$. Notice that all of its edges are inside or on $\hull(U')$, and all remaining segments in
$S\setminus S'$ lie on $\hull(U\cup H)$, see Fig.~\ref{fig:Usteps}. So far, the size of the planar graph is $\bigo(|U|)$, its faces are small, and the
nodes are all points from $U$ and possible some points from $H$, and any segment from $S'$ is an edge there. Therefore, we just need to make
sure that any remaining segment is either an edge, or lies within a face.

\begin{figure}[t]
\centering
\def \svgwidth{0.94\columnwidth}
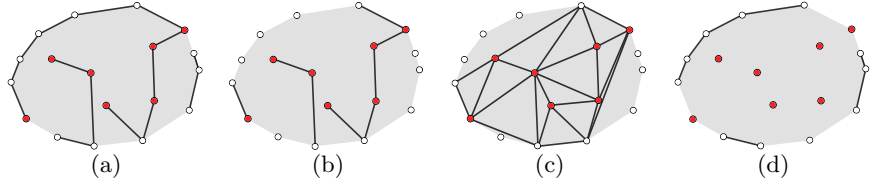
\caption{(a) The segments in $S$, (b) the segments in $S'$, (c) the initial triangulated planar graph, (d) the remaining segments. Points from $U$ are filled.}
\label{fig:Usteps}
\end{figure}

\begin{figure}[b]
\centering
  \centering
  \def \svgwidth{0.6\columnwidth}
  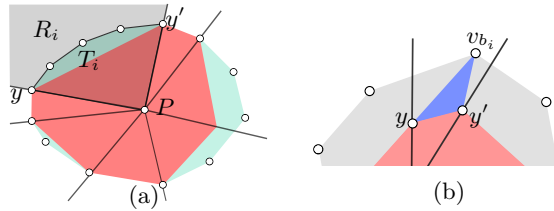
  \caption{(a) The intersection $T_{i}$, (b) adding one triangle when $a_i=b_i$.}
    \label{fig:problem13}
\end{figure}

To deal with the remaining segments, we add $I_{1}$ and $I_{2}$ to the set of nodes. Fix any point $P$ strictly inside $\hull(U')$ and, for every
node $P'$ of $\hull(U')$, draw a ray starting in $P$ and going through $P'$. All these rays partition the region outside $\hull(U')$ into convex subregions
$R_{1},R_{2},\ldots,R_{|\hull(U')|}$. The intersection of $R_{i}$ with $\hull(U\cup H)$, called $T_{i}$, contains exactly two vertices of
$\hull(U')$, call them $y_{i}$ and $y'_{i}$, see Fig~\ref{fig:problem13}(a). We will process every such $T_{i}$ separately, extending the current graph
by adding new triangles.
Consider the sequence of points $v_{a_{i}},v_{a_{i}+1},\ldots,v_{b_{i}}$ of $\hull(U\cup H)$, which belong to $T_{i}$. If the sequence is empty,
there is nothing to do. Otherwise we have two cases:
\begin{enumerate}
\item if $a_{i} =  b_{i}$, create a new triangle $\bigtriangleup v_{b_{i}}y_{i}y'_{i}$ to $\mathcal{T}$, see Fig.~\ref{fig:problem13}(b),
\item if $a_{i} \neq b_{i}$, create two new triangles
$\bigtriangleup v_{a_{i}}y_{i}y'_{i}$,  $\bigtriangleup v_{a_{i}}v_{b_{i}}y'_{i}$. Then add a triangle $\bigtriangleup I_{j}v_{a_{i}}v_{b_{i}}$ if both $v_{a_{i}}$  and $v_{b_{i}}$
belong to the same  $V_{j}$, see Fig.~\ref{fig:problem5}(a). Otherwise either $v_{up}$ or $v_{down}$ is in $v_{a_{i}}, \dots, v_{b_{i}}$, and we add three triangles as in Fig.~\ref{fig:problem5}(b).
\end{enumerate}

\begin{figure}[t]
\centering
\def \svgwidth{1\columnwidth}
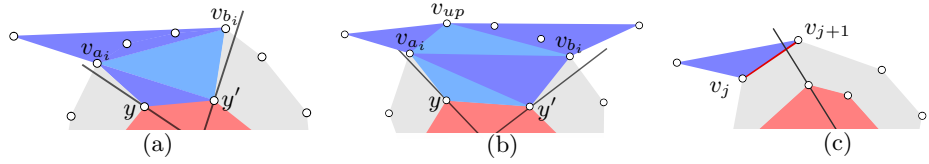
\caption{(a) $v_{a_{i}}, v_{b_{i}}$ in the same $V_j$, (b) $v_{up}$ between $v_{a_{i}}$ and $v_{b_{i}}$, (c) a remaining segment.}
\label{fig:problem5}
\end{figure}

Now any remaining segment which lies within a single $T_{i}$ is inside one of the new triangles. To deal with the other remaining segments, for
each such segment $v_{j}v_{j+1}$ we simply add either $\bigtriangleup I_{1}v_{j}v_{j+1}$ or $\bigtriangleup I_{2}v_{j}v_{j+1}$, see Fig.~\ref{fig:problem5}(c).
This is correct because any two consecutive points on $\hull(H)$ always either both belong to $V_{1}$ or both belong to $V_{2}$.
One ray can cross at most one edge, so the number of triangles created in this step is at most $|U'|$.

By the construction, the insides of any new triangles are disjoint. Also, they all lie outside $\hull(U')$. Hence we can add the new triangles to the
initial planar graph to form a larger planar graph. Because we created $\bigo(U')$ new triangles, the size of the new planar graph is still
$\bigo(U)$, though. Now some of its faces might be large, though, so we include $I_{1}, I_{2}, v_{up}, v_{down}$ in its set of nodes, and all $I_{1}v_{up}$, $I_{1}v_{down}, I_{2}v_{up}, I_{2}v_{down}$ in its set of edges. Finally, we triangulate the large inner faces, if any. The size of the final graph
is $\bigo(U)$ and we ensured that any segment from $S$ is either among its edges, or lies within one of its faces.
\qed
\end{proof}

Lemma~\ref{small_triang} shows that it is indeed enough to iterate over separators of size $c\sqrt{n+2m+2}$, hence Algorithm~\ref{algorithm}
is correct. The remaining part is to argue that it needs just linear space. By Lemma~\ref{GETSPH}, the reduction in line~\ref{line:reduce} uses
$\bigo(m+\ell)$ additional space which can be immediately reused.
Iterating through all ordered subsets of size at most $c\sqrt{n+2m+2}$ can be easily done with
$\bigo(\sqrt{n+2m})$ additional space. Bounding the space necessary to iterate over all possibilities of how the solution intersects with the
separator is less obvious, but the same bound can be derived by looking at how the possibilities were counted. The $\bigo(\sqrt{n+2m})$ additional
space must be stored for every recursive call. Additionally, for each call we must store its arguments $V, T$ and $H$, which takes
$\bigo(n+2m+2\ell)$ additional space. As $n+2m$ decreases by a constant factor in every recursive call, the recursion
depth is $\bigo(\log(n+2m))$, which in turn implies $\bigo(n+2m+2\ell\log(n+2m))$ overall space consumption. Even though we always reduce the 
instance so that $\ell \leq n^{2}$, this bound might be superlinear, and we need to add one more trick.

Recall that the hull pairs in the subproblems are disjoint subsets of all hull pairs in the original problem. Hence, instead of copying the hull pairs
to the subproblems, we can store them in one global array. All hull pairs in the current problem are stored in a contiguous fragment there. 
Before the recursive calls, we rearrange the fragment so that the hull pairs which should be processed in both subproblems are, again, stored
in contiguous fragments of the global array. The rearranging can be done in linear space and constant additional space. There is one problem, though.
When we return from the subproblems, the fragment containing the hull pairs might have been arbitrarily shuffled.
This is a problem, because we are iterating over the ordered subsets of all points, which requires operating on their
indices. Now the order of the hull pairs might change, so we cannot identify a hull point by storing the index of its pair. Nevertheless, we can
maintain an invariant that all hull pairs in the current problem are lexicographically sorted. In the very beginning, we just sort the global array.
Then, before we recurse on a subproblem, we make sure that its fragment is sorted. After we are done with both subproblems, we re-sort the fragment
of the global array corresponding to the current problem. This doesn't increase the total running time and decreases the overall space
complexity to linear.

Together with Theorem~\ref{GETSP}, this gives the final result.

\begin{theorem}
$\ETSP$ can be solved in $\bigo(nk^2+k^\bigo(\sqrt{k}))$ time and linear space.
\end{theorem}

\bibliographystyle{splncs03}
\bibliography{biblio}

\begin{thebibliography}{1}
\providecommand{\url}[1]{\texttt{#1}}
\providecommand{\urlprefix}{URL }

\bibitem{median}
Blum, M., Floyd, R., Pratt, V., Rivest, R., Tarjan, R.: Time bounds for
  selection. Journal of Computer and System Sciences  7,  448--461 (1972)

\bibitem{parametrized}
Deineko, V.G., Hoffmann, M., Okamoto, Y., Woeginger, G.J.: The traveling
  salesman problem with few inner points. Operations Research Letters  34(1),
  106--110 (2006)

\bibitem{FredmanFibonacci}
Fredman, M.L., Tarjan, R.E.: Fibonacci heaps and their uses in improved network
  optimization algorithms. J. ACM  34(3),  596--615 (Jul 1987)

\bibitem{separators}
Hwang, R., Chang, R., Lee, R.: The searching over separators strategy to solve
  some {NP}-hard problems in subexponential time. Algorithmica  9(4),  398--423
  (1993)

\bibitem{kann1992approximability}
Kann, V.: On the Approximability of NP-complete Optimization Problems.
  Trita-NA, Royal Institute of Technology, Department of Numerical Analysis and
  Computing Science (1992)

\bibitem{fasterparametrized}
Knauer, C., Spillner, A.: A fixed-parameter algorithm for the minimum weight
  triangulation problem based on small graph separators. In: Proceedings of the
  32nd international conference on Graph-Theoretic Concepts in Computer
  Science. pp. 49--57. WG'06, Springer-Verlag, Berlin, Heidelberg (2006)

\bibitem{miller}
Miller, G.L.: Finding small simple cycle separators for 2-connected planar
  graphs. In: Proceedings of the sixteenth annual ACM symposium on Theory of
  computing. pp. 376--382. STOC '84, ACM, New York, NY, USA (1984)

\bibitem{npcomplete}
Papadimitriou, C.H.: The {E}uclidean travelling salesman problem is
  {NP}-complete. Theoretical Computer Science  4(3),  237 -- 244 (1977)

\bibitem{Smith}
Smith, W.D.: Studies in Computational Geometry Motivated by Mesh Generation.
  Ph.D. thesis, Princeton University, Princeton, NJ, USA (1989)

\end{thebibliography}

\end{document}